\documentclass[journal]{IEEEtran}

\usepackage{amsmath,bbold, dsfont,yfonts,amssymb,epsfig,amsthm,bm,multirow, graphicx,color, mathrsfs, bbold}
\usepackage{stfloats}

\IEEEoverridecommandlockouts

\usepackage{color}
\usepackage{algpseudocode}
\usepackage{algorithm,eqparbox,array}
\usepackage{float}
\newfloat{algorithm}{t}{lop}

\begin{document}

\title{Leveraging One-hop Information in Massive MIMO Full-Duplex Wireless Systems}

\author{{Wenzhuo Ouyang, Jingwen Bai, Ashutosh Sabharwal, \emph{Fellow, IEEE}}
\vspace{-15pt}

\thanks{The authors are with the Department of ECE, Rice University (e-mails: $\{$wenzhuo.ouyang, jingwen, ashu$\}$@rice.edu).  The work of all three authors was partially supported by NSF Grants CNS-1161596 and CNS-1314822.}}

\maketitle
\begin{abstract}
We consider a single-cell massive MIMO full-duplex wireless communication system, where the base-station (BS) is equipped with a large number of antennas. We consider the setup where the single-antenna mobile users operate in half-duplex, while each antenna at the BS is capable of full-duplex transmissions, i.e., it can transmit and receive simultaneously using the same frequency spectrum. The fundamental challenge in this system is \emph{intra-cell} inter-node interference, generated by the transmissions of uplink users to the receptions at the downlink users. The key operational challenge is estimating and aggregating inter-mobile channel estimates, which can potentially overwhelm any gains from full-duplex operation. 

In this work, we propose a scalable and distributed scheme to optimally manage the inter-node interference by utilizing a \emph{``one-hop information architecture''}. In this architecture, the BS only needs to know the  signal-to-interference-plus-noise ratio (SINR) from the downlink users. Each uplink user needs its own SINR, along with a weighted signal-plus-noise metric from its one-hop neighboring downlink users, which are the downlink users that it interferes with. The proposed one-hop information architecture does not require any network devices to comprehensively gather the vast inter-node interference channel knowledge, and hence significantly reduces the overhead. Based on the one-hop information architecture, we design a distributed power control algorithm and implement such architecture using overheard feedback information. We show that, in typical asymptotic regimes with many users and antennas, the proposed distributed power control scheme improves the overall network utility and reduces the transmission power of the uplink users.
\end{abstract}





\section{Introduction}

With each generation of wireless standards, the number of antennas at the infrastructure nodes has continued to grow to meet the increasing demand of mobile data. For example, both cellular and WiFi standards now support up to 8 antennas at infrastructure nodes. A promising and now standardized approach to use multiple antennas at the infrastructure is to use Multi-User Multiple Input Multiple Output (MU-MIMO) for supporting multiple uplink or downlink data streams in same time-slot. More recently, massive MIMO regime has been explored, which allows for a large number of antennas to reside at the infrastructure~(with orders of magnitude more antennas compared to
conventional MIMO systems), has attracted large interests in both academia and industry. The theoretical~(\cite{Guthy13,Marzetta10,Ngo2013}) and experimental results~(\cite{shepard2012argos,Suzuki12}) have shown that  massive MIMO can lead to a number of desirable properties from the network design point of view, that include reduced inter-beam interference, improved energy efficiency, and reduced inter-cell interference, among others. As a result, massive MIMO is being considered for  standardization as one of the key technologies in
next-generation of wireless systems in both cellular and WiFi~(e.g.,  3GPP LTE Release 12~\cite{3gppMassiveMIMO} and 802.11-ax~\cite{wifiMassiveMIMOFD}). Channel models with 64 BS antennas have already been standardized~\cite{3gppMassiveMIMO}, with research platform supporting even larger configurations; e.g. 96 BS antennas~\cite{Shepard2013}.


In addition to MU-MIMO, another potential avenue to achieve higher spectral efficiency is to leverage full-duplex transmission, where a full-duplex-capable device can transmit and receive at the same time using the same frequency spectrum. The full-duplex mode of network operation has the potential to double the spectrum efficiency of wireless networks and can bring substantial flexibility to higher layer design~\cite{FDapplication}. In fact,  in-band full-duplex has already become part of the ongoing standard both in 3GPP~\cite{3gppFD}  and 802.11-ax~\cite{wifiMassiveMIMOFD}. While it can be hard to integrate the full-duplex capabilities to client devices due to the processing and energy constraint, it is in fact feasible to design near-perfect full-duplex base-stations thanks to available freedom (bigger size, non-battery-powered operation) in their designs~(e.g., see \cite{duarte2014design,Everett13Paper,ashu13} and the references therein).  One method to leverage full-duplex infrastructure with half-duplex mobile handsets is to have \emph{simultaneous} multi-user up- and downlink transmissions. However, the potential for simultaneous up- and downlink MU-MIMO at the BS leads to a new challenge -- the inter-node interference \emph{within} each cell, i.e intra-cell interference, as the transmissions of uplink interfere with the receptions of downlink as illustrated in Figure~\ref{fig:FD_MIMO}. The inter-node interference hence poses a fundamental challenge to enabling full-duplex at the BS and needs to be managed efficiently.

In the presence of intra-cell interference, a centralized scheme can be used where the infrastructure device aggregates comprehensive knowledge of all the uplink, downlink and inter-node channel to make jointly optimal decisions on resource allocation. The centralized scheme may be useful for sparse scenario with small number of users. It will, however, incur a significant amount of overhead as the number of users grows, with the inter-node interference being \emph{the dominant bottleneck}. For example, consider the setup with $M$ full-duplex-capable antennas\footnote{A full duplex capable antenna is capable of transmitting and receiving at the same time. It can be realized by either using one antenna (e.g., \cite{bharadia2013full}) or using a pair of antennas (e.g., \cite{duarte2014design}\cite{Everett13Paper}).} at the BS, and with $K_{ul}$ half-duplex single antenna uplink and $K_{dl}$ half-duplex single antenna downlink users, henceforth denoted as $(M,K_{ul},K_{dl})$ massive-MIMO full-duplex system, in Figure~\ref{fig:FD_MIMO}. The centralized control scheme requires the full-duplex BS to aggregate knowledge of $K_{ul} \times K_{dl}$ inter-node interference channels, in addition to $M \times K_{ul}$ uplink and $M \times K_{dl}$ downlink channels. Note that there has been recent development on uplink/downlink channel estimation in Massive MIMO with reasonable complexity \cite{Ma14,Yin13}. For example, the $M \times K_{ul}$ uplink channels can be estimated by letting $K_{ul}$ uplink users send pilots to the BS, and the downlink channels can be obtained by letting $K_{dl}$ downlink users send pilots and estimated by utilizing channel reciprocity. Hence the $M \times (K_{ul}+K_{dl})$ up- and downlink channels can be obtained with $(K_{ul}+K_{dl})$ transmissions of overhead (pilots). However, to obtain $K_{ul} \times K_{dl}$ inter-node interference channels, each of the $K_{ul}\times K_{dl}$ channel information needs to estimated at the uplink/downlink users and sent to the infrastructure; the inter-node interference channels are two-hop away from the BS. Hence acquiring all the channel knowledge of $K_{ul} \times K_{dl}$ inter-node interference channels incurs overhead at least at the order of $K_{ul} \times K_{dl}$.  As an example, under the scenario with $K_{ul}=15$ uplink and $K_{dl}=20$ downlink users, there can be as many as $K_{ul}\times K_{dl}=300$ inter-node interference channels, which is significantly more than $K_{ul}+ K_{dl}=35$. Therefore, for centralized scheme, obtaining the inter-node interference becomes the main bottleneck and will render the overall centralized architecture non-scalable.

In this paper, we thus seek distributed methods for resource allocation that significantly reduce \emph{who} needs to know \emph{what} about \emph{which} channels. We propose, to the best of our knowledge, the first scalable and distributed interference management techniques in massive MIMO full-duplex systems with half-duplex clients. Our contributions are listed as follows.

\begin{itemize}
\item We create a scalable information architecture, namely the \emph{one-hop information architecture}, that specifies the network information needed in each cellular devices to perform optimal interference management. Specifically, we consider an architecture where (i)~BS knows only SINR of downlink channels, much like current cellular systems which rely on periodically collected SINR information for power and rate control, and (ii)~uplink transmitters know only about interference channel strength, along with a weighted interference-plus-noise (IN) metric, from the nearby downlink receivers. Therefore, for the $(M,K_{ul},K_{dl})$ massive-MIMO full-duplex system, compared to a fully centralized decision, the proposed system requires only local, at most $K_{dl}$ one-hop information at each uplink user from its neighboring downlink receivers, thus avoiding the significant non-scalable overhead for aggregating the $K_{ul}\times K_{dl}$ interference channel information at the BS.

\item Based on the one-hop information architecture, we propose a distributed power allocation algorithm to achieve the optimal sum-utility in the system. By analyzing the structural properties of the $(M,K_{ul},K_{dl})$ massive-MIMO full-duplex system, we identify that the up- and downlink power control, which is coupled by inter-node interference, can in fact be carefully decoupled. Such a decoupled structure turns out to be an enabler for the one-hop information architecture and the optimal distributed algorithm.

\item We further show that, in typical asymptotic regimes where the amount of antennas scales up with the number of users, distributed power allocation can not only improve the overall system utility, it  can also asymptotically reduce the overall amount of transmission power and hence improve the energy efficiency of mobile users.
\end{itemize}

Note that, there can be scenarios with dense users and hence strong inter-node interference, where half duplex mode can be more efficient solution. The scope of this paper is not to specify when and where to use full duplex mode of operation, but rather to explore the optimal and scalable power control schemes for massive MIMO when the full duplex mode is in use.  

The rest of the paper is organized as follows. In Section~\ref{sec:System} we formulate the physical layer model of massive MIMO full-duplex system. We proceed to present the power control problem in Section~\ref{sec:pwr_obj}. In Section~\ref{sec:dist_alg}, we present the one-hop information architecture that facilitates us to propose distributed interference management algorithm. An overhearing-based scheme is presented in Section~\ref{sec:overhearing} to implement the information architecture. We present asymptotic analysis in large number of antenna regime in Section~\ref{sec:asymp}. Numerical evaluation are provided in Section~\ref{sec:numerics} followed by conclusion remarks in Section~\ref{sec:conclusion}.

\begin{figure}
\centering
   \includegraphics[width=2in,angle=90]{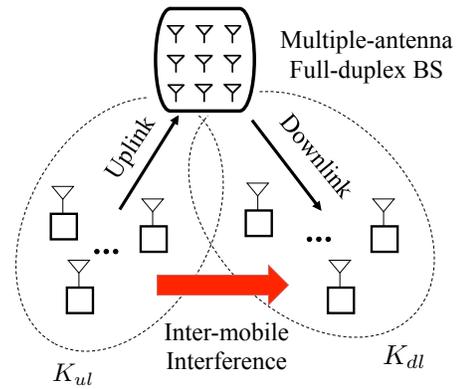}
\caption{Full-duplex multi-user network with $K_{ul}$ uplink users transmitting to the BS and $K_{dl}$ downlink users receiving from BS.}
\label{fig:FD_MIMO}
\vspace{-5pt}
\end{figure}

\section{Related Works}

There has been extensive research on resource allocation in various interference-constrained wireless networks. The interference in the wireless networks is often captured by the SINR, which in turn determines the achievable transmission rate. Power control is an important mechanism to mitigate network interference and have been well studied (see survey \cite{Mung_survey}). Such problem often requires solving nonconvex optimization problem that is proven to be NP hard to solve \cite{luo2008dynamic}. 

One line of research has focused on centralized control where a centralized controller gathers all the information across the network to perform resource allocation and interference management. For instance, the authors in \cite{julian2002qos} proposed centralized power control algorithm under high SINR regime using geographic programming. Signomial programming based approaches are proposed in \cite{chiang2007power} to achieve a local optimal solutions. An iterative optimal algorithm was proposed in \cite{MapleHuang} to achieve global optimality for a subset of non-convex objective functions. A branch and bound algorithm was proposed in \cite{weeraddana2011weighted} to solve weighted sum-rate maximization problems. Approximation algorithm was proposed in \cite{tan2013fast} to solve the NP-hard power control problem.


Another line of research has studied distributed power control algorithms for interference management. In the high SINR regime, distributed power control algorithm was converted to geometric programming \cite{chiang2005balancing}. Distributed price-based power control algorithm was proposed in \cite{huang2006distributed} and shown to be optimal in certain scenarios. Distributed Power Control and scheduling based on back-pressure algorithm is proposed in \cite{xi2010throughput} in high SINR regime. Recently, optimal distributed scheduling based on SINR model is proposed in to achieve transmission rate over a \emph{power control region} \cite{Proutiere2013}. 

Some recent work studied resource allocation for full duplex systems. The authors in \cite{di2014radio} studied optimal subcarrier allocation for Full Duplex OFDMA network. A cell partitioning method was proposed in \cite{shao2014analysis} that allocates frequency resources according to the partitioning of a single cell. 



\section{System Model}
\label{sec:System}
In this section, we describe the channel model for a single-cell massive MIMO full-duplex system where the BS is full-duplex capable while the mobile users operates in single antenna half-duplex mode.  We assume the full-duplex BS is equipped with $M$ antennas to serve \emph{both} uplink and downlink users in the same time/frequency slot. The sets of uplink and downlink users are denoted by $\mathcal{K}_{ul}$ and $\mathcal{K}_{dl}$ respectively, with $|\mathcal{K}_{ul}|={K}_{ul}$ and $|\mathcal{K}_{dl}|={K}_{dl}$, as shown in Figure~\ref{fig:FD_MIMO}.

For each uplink/downlink user $k$, an associated utility function $U_k(r)$ represents the utility it achieves under transmission rate $r$. It reflects the level of satisfaction achieved for each of the users at a given rate. Each utility function is assumed to be a non-decreasing, continuously differentiable, strictly concave function, where $U_k'(r)$ is the derivative of the utility function. The concavity property captures the diminishing return property of the utility achieved by the users in the network \cite{shakkottai2008network}. These utility functions are also used to capture different fairness requirement over the network. We are interested in controlling the transmission power, over both uplink and downlink in the massive MIMO full-duplex system, so that the achievable transmission rate provides the maximum overall sum-utility over the network. To characterize the physical-layer relationship between the transmission power and the corresponding achievable transmission rate, we start by presenting the channel model of the massive MIMO full-duplex system.

The propagation channel model in our system consists of two parts: small-scale fading due to multipath  and large-scale fading due to path loss. The channels are represented by the set $\mathcal{G}=\{\mathbf{G}^{ul},\mathbf{G}^{dl},\mathbf{G}^I\}$, where each element in $\mathcal{G}$ denotes the propagation matrix of complex-valued channel coefficients for uplink, downlink and inter-node interference channels, respectively. We can express $\mathbf{G}^k,~k\in\{ul,dl,I\}$ as
\[  \mathbf{G}^k=\mathbf{L}^{k}(\mathbf{D}^{k}_{\mathbf{g}})^{\frac{1}{2}}
\]
where $\bold{L}^{k}$ represents the small-scale fading matrix with a magnitude of one for each entry, and $(\mathbf{D}^{k}_{\mathbf{g}})^\frac{1}{2}=diag(\sqrt{\mathbf{g}^k})$ is a diagonal matrix whose diagonal entries constitutes a vector $\mathbf{g}^k$ of path loss coefficients.
  
\subsection{Uplink}

At the BS with $M$ antennas, the received uplink signal $\mathbf{y}^{ul}$ at the BS is an ${{M}\times 1}$ vector given as
\begin{gather}
\begin{aligned}
\mathbf{y}^{{ul}}&=\mathbf{G}^{ul}\mathbf{x}^{{ul}}+\mathbf{z}^{{ul}},
\end{aligned}
\end{gather}
where $\mathbf{G}^{ul}$ is an $M\times K_{ul}$ uplink channel matrix and $\mathbf{x}^{ul}$ is a $K_{ul}\times 1$ vector, with each element representing the transmit signal of the uplink users. The vector $\mathbf{z}^{{ul}}$ represents the receiver additive noise, each entry is independent and identically distributed~(i.i.d.) and drawn from a circularly symmetric complex Gaussian distribution with zero mean and variance of $N_0$ denoted as $\mathcal{C}\mathcal{N}(0,N_0)$.

In the uplink, we let $\mathbf{P}^{ul}=[P^{ul}_1,P^{ul}_2,\cdots,P^{ul}_{K_{ul}}]$ denote the vector of transmit powers, where the $P^{ul}_i$ represents the transmit power of uplink user $i$. Each user is subject to a maximum instantaneous power constraint $P_i^{max}$, i.e., $P^{ul}_i\leq P^{ul}_{max}~,\forall i\in \mathcal{K}_{ul}$. 

Following the techniques in \cite{Marzetta10,rusek2013scaling},
 the BS can simply process the received signal using maximum-ratio combining~(MRC), 
\begin{gather}
(\mathbf{G}^{ul})^H\mathbf{y}^{ul}=(\mathbf{G}^{ul})^H \mathbf{G}^{ul}\mathbf{x}^{{ul}}+(\mathbf{G}^{ul})^H\mathbf{z}^{{ul}}, \label{ul}
\end{gather}
where superscript $H$ denotes conjugate transpose. We hence have 
\begin{gather}
\frac{(\mathbf{G}^{ul})^H \mathbf{G}^{ul}}{M}=(\mathbf{D}^{ul}_{\mathbf{g}})^\frac{1}{2}\left(\frac{(\mathbf{L}^{ul})^H\mathbf{L}^{ul}}{M}\right)(\mathbf{D}^{ul}_{\mathbf{g}})^\frac{1}{2},
\end{gather}
where $(\mathbf{D}^{ul}_{\mathbf{g}})^\frac{1}{2}$ is a $K_{ul}\times K_{ul}$ matrix and $\mathbf{L}^{ul}$ is an $M\times K_{ul}$ matrix.
In massive MIMO, the number of antennas at BS is much larger than that at mobiles. Thus when $M\gg K_{ul}$, under favorable propagation conditions as manifested in~\cite{Marzetta10}, the columns of the small-scale fading matrix are asymptotically orthogonal, the following convergence is obtained~\cite{Marzetta10}, 
\begin{gather}
\left(\frac{(\mathbf{G}^{ul})^H \mathbf{G}^{ul}}{M}\right)_{M\gg K_{ul}}\rightarrow \mathbf{D}^{ul}_{\mathbf{g}} \label{mm}
\end{gather}
Now substituting~(\ref{mm}) into (\ref{ul}), we have
\begin{gather}
\left(\frac{(\mathbf{G}^{ul})^H\mathbf{y}^{ul}}{M}\right)_{M\gg K_{ul}}\rightarrow \mathbf{D}^{ul}_{\mathbf{g}}\mathbf{x}^{{ul}}+\frac{(\mathbf{G}^{ul})^H\mathbf{z}^{{ul}}}{M}. \label{eq:ortho_ul}
\end{gather}

We then compute the Signal-to-Interferene-and-Noise ratio ($\mathsf{SINR}$) for uplink user $i$ as 
\begin{gather}
\mathsf{SINR}^{ul}_i=\frac{M P^{ul}_i g^{ul}_i}{N_0},\label{eq:SINR_UL}
\end{gather}
where $\sqrt{g^{ul}_i}$ is the $i$th entry in the diagonal matrix $(\mathbf{D}^{ul}_{\mathbf{g}})^\frac{1}{2}$. Note that in the above expression, we have used the asymptotic result in~(\ref{eq:ortho_ul}). Under the domain of massive MIMO, we assume the SINR expression is valid for large value of $M$, hence ignoring the negligible error term.

Therefore, with receiver MRC, the uplink user $i$ achieves a transmission rate
\begin{align}
R_i^{ul}(P^{ul}_i)=\log\Big(1+\frac{M\cdot P^{ul}_i g_i^{ul}}{N_0}\Big), i\in\mathcal{K}_{ul},\label{eq:ul_rate}
\end{align}
where, recall that, $g_i^{ul}$ is the corresponding path loss coefficient for $i$-th user.

\subsection{Downlink}
The received signal at the $K_{dl}$ downlink users $\mathbf{y}^{{dl}}$ is a combination of the downlink signals and the interfering uplink signals, and is given by
\begin{gather}
\begin{aligned}
\mathbf{y}^{{dl}}&=\mathbf{G}^{dl}\mathbf{x}^{{dl}}+\mathbf{G}^{I}\mathbf{x}^{{ul}}+\mathbf{z}^{{dl}}, \label{dl}
\end{aligned}
\end{gather}
where $\mathbf{y}^{{dl}}$ is a $K_{dl}\times 1$ received signal vector, $\mathbf{G}^{dl}$ is a $K_{dl}\times M$ downlink channel matrix and $\mathbf{x}^{dl}$ is an $M\times 1$ downlink transmit signal vector; $\mathbf{G}^{I}$ is a $K_{dl}\times K_{ul}$ inter-node interference channel matrix; $\mathbf{z}^{{dl}}$ is the receiver additive Gaussian noise which contains i.i.d. $\mathcal{C}\mathcal{N}(0,N_0)$ entries.

In the downlink, we let $\mathbf{P}^{dl}=[P^{dl}_1,P^{dl}_2,\cdots,P^{dl}_{K_{dl}}]$ denote the vector where the element $P^{dl}_j$ represents the transmit power to downlink user $j$. The downlink transmission power allocation is subject to a constraint on the total amount of transmission powers, i.e., $\sum_{j=1}^N P^{dl}_j=P^{dl}_{tot}$. 

Note that we assume perfect serf-interference cancellation at the full-duplex-capable BS. This is inline with recent results that self-interference cancellation is capable of suppressing self-interference below the noise floor \cite{Everett13Paper}.

Similar to the uplink, in massive MIMO downlink,  under the conditions of $M\gg K_{dl}$, we have
\begin{gather}
\left(\frac{\mathbf{G}^{dl}(\mathbf{G}^{dl})^H}{M}\right)_{M\gg K_{dl}}\rightarrow \mathbf{D}^{dl}_{\mathbf{g}}. \label{mmdl}
\end{gather}

Following the techniques in \cite{rusek2013scaling}, BS can perform conjugate beamfoming precoding,
\begin{gather}
\begin{aligned}
\mathbf{x}^{{dl}}=\frac{1}{\sqrt{M}}(\mathbf{G}^{dl})^H  (\mathbf{D}^{dl}_{\mathbf{g}})^{-\frac{1}{2}} \mathbf{D}_p^\frac{1}{2} \mathbf{s}^{dl}. \label{eq:dl_symbol}
\end{aligned}
\end{gather}
where $\mathbf{D}_p^\frac{1}{2}=diag(\sqrt{\mathbf{P}^{dl}})$ is a $K_{dl}\times K_{dl}$ diagonal matrix. The $K_{dl}\times 1$ vector $\mathbf{s}^{dl}$ contains downlink messages where each message has unit power, i.e., $E(|s_j^{dl}|)=1$ for all $j\in\mathcal{K}_{dl}$. 

We substitute (\ref{mmdl}) and (\ref{eq:dl_symbol}) into (\ref{dl}) and we obtain that
\begin{gather}
\left(\frac{\mathbf{y}^{{dl}}}{\sqrt{M}}\right)_{M\gg K_{dl}}\rightarrow (\mathbf{D}_{\mathbf{g}}^{dl})^{\frac{1}{2}} \mathbf{D}_p^\frac{1}{2}\mathbf{s}^{{dl}}+\frac{\mathbf{G}^{I}\mathbf{x}^{{ul}}}{\sqrt{M}}+\frac{\mathbf{z}^{{dl}}}{\sqrt{M}}.
\end{gather}

Because the uplink and downlink transmissions are over the same frequency band, the uplink transmissions causes interference to the downlink receptions, as shown in Figure~\ref{fig:FD_MIMO}.  The interference channel strength between uplink $i$ and downlink $j$ is denoted as $g^I_{ij}$. For the uplink user $i$, we let $\mathcal{N}_i\subset \mathcal{K}^{dl}$ denote the set of downlink users that receives non-negligible interference from user $i$. For downlink user $j$, we let $\mathcal{N}_j\subset \mathcal{K}^{ul}$ denote the set of uplink users that imposes non-negligible interference to user $j$.

In the massive MIMO scenario where $M\gg K_{dl}$, the $\mathsf{SINR}$ of downlink user $j$ can then be calculated as
\begin{align}
\mathsf{SINR}^{dl}_j=\frac{M\cdot P^{dl}_j g_j^{dl}}{\sum_{i\in \mathcal{N}_j}g^I_{ij}P^{ul}_i +N_0}\label{eq:SINR_DL}
\end{align}
where $\sqrt{g^{dl}_j}$ is the $j$th entry in the diagonal matrix $(\mathbf{D}^{dl}_{\mathbf{g}})^\frac{1}{2}$, and, similar to the uplink, we have ignored the negligible error term in~(\ref{eq:SINR_DL}).

We consider the setup where the downlink receivers treat the inter-node interference as noise. Therefore, with transmit conjugate beamforming, the downlink transmission rate to user $j\in\mathcal{K}_{dl}$ is expressed as
\begin{align}
R_j^{dl}({P}_j^{dl}, \mathbf{P}^{ul})=\log\left(1+\frac{M\cdot P^{dl}_j g_j^{dl}}{\sum_{i\in \mathcal{N}_j}g^I_{ij}P^{ul}_i +N_0}\right),\label{eq:dl_rate}
\end{align}
which depends on both uplink and downlink power allocation schemes $\mathbf{P}^{ul}, \mathbf{P}^{dl}$. We henceforth denote the Interference-plus-Noise (IN) by $\mathsf{IN}_j=\sum_{i\in \mathcal{N}_j}g^I_{ij}P^{ul}_i +N_0$.

\section{Power Control Problem}
\label{sec:pwr_obj}

In this section, we formally define the power control problem. The objective of power allocation is to control the transmission powers $\mathbf{P}^{ul}, \mathbf{P}^{dl}$ to jointly manage the inter-node interference and intelligently allocate the downlink transmission powers under the interference, so that the overall sum-utility in the system is maximized. Recall that $R_i^{ul}(P^{ul}_i)$ and $R_j^{dl}(P_j^{dl},\mathbf{P}^{ul})$ are defined in~(\ref{eq:ul_rate}) and~(\ref{eq:dl_rate}) respectively. The power allocation problem is defined as the following optimization problem.
\begin{align}
&\max_{\mathbf{P}^{ul}, \mathbf{P}^{dl}} \ \sum_{i\in \mathcal{K}_{ul}}U_i\left(R_i^{ul}(P^{ul}_i)\right)+\sum_{j\in \mathcal{K}_{dl}} U_j\left(R_j^{dl}({P}_j^{dl},\mathbf{P}^{ul})\right)\label{eq:obj_accurate}\\
& \ s.t., \hspace{0.2in} \sum_{j\in\mathcal{K}_{dl}} P^{dl}_j\leq P^{dl}_{tot}, \label{eq:DL_sum_accurate}\\
&\hspace{0.36in} 0 \leq P^{ul}_i\leq P^{ul}_{max}, i\in\mathcal{K}_{ul}, P^{dl}_j\geq 0, j\in\mathcal{K}_{dl}.\label{eq:range_accurate}
\end{align}
As previously discussed, the approach of centralized control can be used to solve the problem at the infrastructure, where the full-duplex BS aggregates comprehensive knowledge across the cellular network. To optimally manage the inter-node interference imposed by the uplink users to downlink receivers, the infrastructure needs the knowledge of all the inter-node interference channel gains. From the form of optimization problem, the required amount of information at the BS is listed in Table~\ref{tab:central}. \footnote{Note that in the above table, we do not include the parameters $M$, $P^{dl}_{tot}$,$P^{ul}_{max}$, which are assumed to be predetermined information that does not need to be communicated network wide.}
\begin{table}[h]  
\centering
\begin{tabular}{|c|c|}
\hline
\textbf{Information at BS} & \textbf{Amount} \\ \hline
Uplink path loss: $g_i^{ul}$ & All $i\in\mathcal{K}_{ul}$ \\ \hline
Downlink path loss: $g_j^{dl}$ & All $j\in\mathcal{K}_{dl}$ \\ \hline
Inter-node interference:  $g_{ij}^{I}$ & All $(i,j)$ pairs, $i{\in}\mathcal{K}_{ul},j{\in}\mathcal{N}_i$ \\ \hline 
\end{tabular}
\caption{Required information for centralized algorithm at BS}\label{tab:central}
\end{table}

As discussed in the introduction, the above centralized solution is non-scalable primarily because of the significant overhead to potentially collect $ {K}_{ul}\times {K}_{dl}$ inter-node interference channels. Therefore, our goal is to design a scalable distributed structure, where the network-wide information are not aggregated at the BS. Instead, the information is distributed across the cellular network so that each device only acquires information at most one hop away. 

\section{Distributed Solution using one-hop information}
\label{sec:dist_alg}

In this section, we present a scalable network state information architecture, denoted by one-hop information. Such information architecture significantly reduces the overhead associated with centrally aggregating all the information at the infrastructure. Specifically, the one-hop information architecture is described below in Table~\ref{tab:BS_1hop}-\ref{tab:UL_1hop}, where the Interference-plus-Noise (IN) metric, i.e., $m_{i,j}[t]$ is formally defined in~(\ref{eq:metric}) in Section~\ref{sec:distributed}.
$  $
\begin{table}[h]
\centering
\begin{tabular}{|c|c|}
\hline
\textbf{Information at BS} & \textbf{Amount} \\ \hline
Downlink SINR  $\mathsf{SINR}^{dl}_j[t]$  & All $j\in\mathcal{K}_{dl}$ \\ \hline
\end{tabular}
\caption{One-hop information architecture, BS}\label{tab:BS_1hop}
\end{table}

\begin{table}[h]
\centering
\begin{tabular}{|c|c|}
\hline
\textbf{Information at uplink user $i$} & \textbf{Amount} \\ \hline
Uplink SINR:  $\mathsf{SINR}^{ul}_i[t]$  & 1 \\ \hline
IN metric: $m_{ij}[t]$  & All $j\in\mathcal{N}_{i}$ \\ \hline
\end{tabular}
\caption{{One-hop information architecture, uplink user $i$}}\label{tab:UL_1hop}
\end{table}

It can be observed from the above tables that the BS requires only SINR from the downlink users, which does not incur additional overhead to current cellular systems where BS periodically collects SINR from downlink users for power and rate control. The BS hence does not need to know information two hops away, e.g., inter-node interference. For each uplink user $i$, it only requires the one-hop information, i.e., the IN metric $m_{ij}$ for all users $j\in \mathcal{N}_i$. It does not require information more than one-hop away, e.g., $m_{i_0j}[t]$ for all users $i_0\neq i, j\in \mathcal{K}^{dl}$. Therefore, the multiplicative overhead associated centralized solutions disappears. In the the following sections, we show that under massive MIMO systems with high power gain and hence high SINR, the above one-hop information architecture enables design of optimal distributed power allocation scheme to manage inter-node interference.

\subsection{High SINR Model}

In massive MIMO full-duplex system, as illustrated in the rate expressions (\ref{eq:ul_rate}) and (\ref{eq:dl_rate}), large power gain is brought to each user owing to the large number of antennas at the BS. The high power gain in massive MIMO brings high SINR values at each receivers. Under the high SINR, the uplink transmission rates hence can be approximated by
\begin{align}
\widetilde{R}_i^{ul}(P_i^{ul}){=}\log(\mathsf{SINR}^{ul}_i){=}\log\left(\frac{M\cdot P^{ul}_i g_i^{ul}}{N_0}\right), \ i\in\mathcal{K}_{ul},\label{eq:ul_rate_approx}
\end{align}
for $P^{ul}_i\geq P_i^0$ where the parameter $P_0^i$ is the power allocation constraint such that the high SINR approximation is valid, i.e., $\mathsf{SINR}^{ul}_i>>1$. Under massive MIMO setup, $P^0_i, i\in \mathcal{K}_{ul}$ are small values. As $M$ scales further up, the value of $P^0_i$ scales down. 

In the downlink, the SINR expression given in~(\ref{eq:SINR_DL}) is high thanks to the large value of $M$ compared to the number of users in the neighborhood in the range of the cell. Correspondingly, the high SINR model gives the transmission rate expression
\begin{align}
\widetilde{R}_j^{dl}({P}_j^{dl},\mathbf{P}^{ul})=&\log\left(\frac{M\cdot P^{dl}_j g_j^{ul}}{\sum_{i\in \mathcal{N}_j}g^I_{ij}P^{ul}_i +N_0}\right), \ j\in\mathcal{K}_{dl},\label{eq:dl_rate_approx}
\end{align}
which, similar to the uplink, is valid for $P^{dl}_j\geq P_j^0$. The parameter $P_j^0$ is a small value as the number of BS antennas scales up.

Following the techniques in \cite{chiang2005balancing,xi2010throughput}, we perform a change of variables by letting $\widehat{P}^{ul}_i=\log(P^{ul}_i)$ and $\widehat{P}^{dl}_j=\log(P^{dl}_j)$ for all $i\in\mathcal{K}_{ul}$ and $j\in\mathcal{K}_{dl}$. With the new set of variables and under high SINR approximation~(\ref{eq:ul_rate_approx})(\ref{eq:dl_rate_approx}), the utility maximization problem (\ref{eq:obj_accurate})-(\ref{eq:range_accurate}) becomes
\begin{align}
&\max_{\mathbf{r}} \ \sum_{i\in \mathcal{K}_{ul}}U_i({r}_i)+\sum_{j\in \mathcal{K}_{dl}} U_j({r}_j)\label{eq:obj_log}\\
& \ s.t., \hspace{8pt} r_i-\log(M g_i^{ul}){-}\widehat{P}^{ul}_i\leq -\log N_0, \ i\in \mathcal{K}_{ul}\label{eq:ul_constraint_log}\\
&\hspace{0.11in} r_j{-}\log(M g_j^{dl}){-}\widehat{P}^{dl}_j{+}\log\Big(\hspace{-2pt}\sum_{i\in \mathcal{N}_j}g^I_{ij}\exp(\widehat{P}^{ul}_i){+}N_0\Big)\leq 0,\nonumber\\
&\hspace{2.4in}  j\in \mathcal{K}_{dl}\label{eq:dl_constraint_log}\\
&\hspace{0.37in}\sum_{j\in\mathcal{K}_{dl}} \exp(\widehat{P}^{dl}_j)\leq P^{dl}_{tot}\label{eq:dl_sun_pwr_log}\\
&\hspace{0.37in}\log P_i^0\leq \widehat{P}^{ul}_i\leq \log P^{ul}_{max}, i\in \mathcal{K}_{ul}\label{eq:ul_pwr_log}\\
&\hspace{0.37in} \widehat{P}^{dl}_j \geq \log P^0_{j}, j\in \mathcal{K}_{dl}, r_k\geq 0, k\in \mathcal{K}_{ul}\cup\mathcal{K}_{dl}\label{eq:ul_pwr_bound}
\end{align}
where $\mathbf{r}=[r_k]_{k\in \mathcal{K}_{ul}\cup\mathcal{K}_{dl}}$.

\noindent\textbf{Observatoin~1: }The high SINR provided by massive MIMO system enables us to transform the non-convex problem  (\ref{eq:obj_accurate})-(\ref{eq:range_accurate}) to the above convex optimization problem.



\subsection{Leveraging Full-Duplex Architecture to \\Reduce Overhead}

One of the main characteristics and bottleneck of the massive MIMO full-duplex system is the inter-node interference. Such interference not only can incur significant overhead in the centralized solution, but also \emph{couples} up- and downlink power control decisions. Because the uplink transmission power causes interference to downlink reception, the downlink power allocation will also be affected accordingly. In this section, we present an interesting observation that there exists a way to \emph{decouple} the interplay between up- and downlink power allocations, thanks to the uni-directional uplink to downlink interference in full-duplex system. This  structure helps reduce the overhead associated with obtaining network information to achieve the optimal utility for the aforementioned problem~(\ref{eq:obj_log})-(\ref{eq:ul_pwr_bound}).

For each uplink user $i\in \mathcal{K}_{ul}$, we associate a Lagrangian multiplier $q^{ul}_i$ that corresponds to the constraint~(\ref{eq:ul_constraint_log}). A Lagrangian multiplier $q_j^{dl}$ is associated with downlink user $j\in \mathcal{K}_{dl}$ that corresponds to constraint~(\ref{eq:dl_constraint_log}). The following dual function is then obtained
\begin{align}
D(\mathbf{q}^{ul},\mathbf{q}^{dl})=&\max_{\mathbf{r}}B(\mathbf{r},\mathbf{q}^{ul}\hspace{-2pt},\mathbf{q}^{dl}\hspace{-1pt})+\max_{\mathbf{\widehat{P}}^{ul}} V_{ul}(\mathbf{\widehat{P}}^{ul},\mathbf{q}^{ul},\mathbf{q}^{dl})\nonumber\\
&+ \max_{\mathbf{\widehat{P}}^{dl}} V_{dl}(\mathbf{\widehat{P}}^{dl},\mathbf{q}^{dl})-\sum_{i\in\mathcal{K}_{ul}}q_i^{ul}\log N_0.\label{eq:dual}
\end{align}
where
\begin{align}
&B(\mathbf{r},\mathbf{q}^{ul}\hspace{-2pt},\mathbf{q}^{dl}\hspace{-1pt}){=}\hspace{-2pt}\sum_{i\in\mathcal{K}_{ul}}\Big(U_i(r_i){-}q_i^{ul}{\cdot}r_i\Big){+}\hspace{-2pt}\sum_{j\in\mathcal{K}_{dl}}\hspace{-2pt}\Big(U_j(r_j){-}q_j^{dl}{\cdot}r_j\Big),\label{eq:rate_ctrl}\\
&V_{ul}(\mathbf{\widehat{P}}^{ul},\mathbf{q}^{ul},\mathbf{q}^{dl})= \sum_{i\in\mathcal{K}_{ul}}q_i^{ul}\cdot\log(M g_i^{ul})+ \sum_{i\in\mathcal{K}_{ul}}q_i^{ul}\widehat{P}^{ul}_i\nonumber\\
&\hspace{0.1in}-\sum_{j\in\mathcal{K}_{dl}}q_j^{dl}\cdot\log\Big(\sum_{i\in \mathcal{N}_j}g^I_{ij}\exp(\widehat{P}^{ul}_i) +N_0\Big), \nonumber\\
&V_{dl}(\mathbf{\widehat{P}}^{dl},\mathbf{q}^{dl})= \sum_{j\in\mathcal{K}_{dl}}q_j^{dl}\cdot\log(M\cdot g_i^{dl})+\sum_{j\in\mathcal{K}_{dl}}q_j^{dl}\widehat{P}^{dl}_j, \nonumber
\end{align}
with vectors $\mathbf{q}^{ul}=[q_i^{ul}]_{i\in \mathcal{K}_{ul}}$, $\mathbf{q}^{dl}=[q_j^{dl}]_{j\in \mathcal{K}_{dl}}$, $\widehat{\mathbf{P}}^{ul}=[\widehat{{P}_i}^{ul}]_{i\in \mathcal{K}_{ul}}$, $\widehat{\mathbf{P}}^{dl}=[\widehat{{P}}_j^{dl}]_{j\in \mathcal{K}_{dl}}$. 


Hence, given $\mathbf{q}^{ul},\mathbf{q}^{dl}$, the dual problem is decomposed into three sub-problems, which leads to the following immediate observation.

\textbf{Observation 2.} The uplink power control sub-problem $\max_{\mathbf{\widehat{P}}^{ul}} V_{ul}(\mathbf{\widehat{P}}^{ul},\mathbf{q}^{ul},\mathbf{q}^{dl})$, and the downlink power control sub-problem  $\max_{\mathbf{\widehat{P}}^{dl}} V_{dl}(\mathbf{\widehat{P}}^{dl},\mathbf{q}^{dl})$ are \emph{completely decoupled} as separate problems. 

\textbf{Remark:} The Lagrangian approach uses a similar technique compared with other works that considers power control problems in general networks (e.g., \cite{chiang2005balancing}). However, the full-duplex capability at the BS provides a unique decoupling structures for up- and downlink power control. In the full-duplex systems that we have considered, inter-node interference limits system performance, and up- and downlink power controls are coupled because of inter-node interference. The decoupling structure, however, interestingly reveals that these two power control schemes can actually be performed separately, given some pricing information. Such separation structure enables us to significantly reduce the information passing overhead to design the optimal power control policies.

For the uplink and downlink sub-problems, we obtain the gradient expressions as,
\begin{align}
\frac{\partial V_{ul}(\mathbf{\widehat{P}}^{ul},\mathbf{q}^{ul},\mathbf{q}^{dl})}{\partial \widehat{P}^{ul}_i}&=q^{ul}_i{-}\sum_{j: i\in \mathcal{N}_j} q_j^{dl} \frac{g^I_{ij}\exp(\widehat{P}^{ul}_i)}{\mathsf{IN}_j},\label{eq:UL_gradient}\\
\frac{\partial V_{dl}(\mathbf{\widehat{P}}^{dl},\mathbf{q}^{dl})}{\partial \widehat{P}^{dl}_j}&=q_j^{dl}.\label{eq:DL_gradient}
\end{align}

The sub-problem $\max_{\mathbf{r}}B(\mathbf{r},\mathbf{q}^{ul}\hspace{-2pt},\mathbf{q}^{dl}\hspace{-1pt})$ is the primal rate adaptation, which can be solved with ${r}^{*}_i=(U')^{-1}({q}_i^{ul})$, ${r}^{*}_j=(U')^{-1}({q}_j^{dl})$ where $i\in\mathcal{K}_{ul}, j\in\mathcal{K}_{dl}$. Hence the Lagrangian multipliers can be considered as `price' for achieving a unit of up- and downlink transmission rate.

We next design a distributed power control algorithm based on the aforementioned intuitions.

\subsection{Power Control using one-hop Information}
\label{sec:distributed}

We formally present the distributed power control and interference management algorithm next, where the notation $[x]_{\mathcal{S}}$ stands for projecting the value $x$ onto set $\mathcal{S}$, and the IN metric is 
\begin{align}
m_{ij}[t]=\frac{q^{dl}_j[t]g^I_{ij}}{\mathsf{IN}_j[t]}.\label{eq:metric}
\end{align}

\renewcommand{\thealgorithm}{}

\algnewcommand{\algorithmicgoto}{\textbf{go to}}%
\algnewcommand{\Goto}[1]{\algorithmicgoto~\ref{#1}}%

\vspace{-10pt}\begin{algorithm}[!htp]
  \caption{Distributed Algorithm for Massive MIMO Full-duplex Systems based on One-Hop Information}\label{alg:MIMO_FD}
  \begin{algorithmic}[1]
  \State \textbf{Initialization phase.} At $t=0$, set $q^{ul}_i[0]=0$, $q^{dl}_j[0]=0$, $\widehat{P}_i^{ul}[0]=\log {P}_{max}^{ul}$, $\widehat{P}_j^{dl}[0]=\log (P^{dl}_{tot}/{K}_{dl})$, $\mathsf{IN}_j[0]=N_0$ for all $i\in\mathcal{K}_{ul}$, $j\in\mathcal{K}_{dl}$. Go to step 5.
  \State \textbf{Downlink power update at the BS.} The BS updates the transmission power to the downlink users according to the gradient~(\ref{eq:DL_gradient}),
  \begin{align}
&\widehat{P}_j^{dl}[t]\nonumber\\
=&\Big[\widehat{P}_j^{dl}[t{-}1]+\gamma\cdot q_j^{dl}[t{-}1]\Big]\Big]_{\mathbf{1}\cdot {P}_j^{dl}[t-1]\leq P^{dl}_{tot},P^{dl}_j[t{-}1]\geq P_j^0},\nonumber
\end{align}
where $\gamma$ is the step size.
  \State \textbf{Uplink power update at the users.} The uplink users update the transmission power to the BS according to the gradient~(\ref{eq:UL_gradient}),
  \begin{align}
&\widehat{P}_i^{ul}[t]{=}\Big[\widehat{P}_i^{ul}[t{-}1]{+}\gamma\cdot\big[q_i^{ul}[t{-}1]\nonumber\\
&{-}\sum_{j\in\mathcal{N}_i}m_{ij}[t{-}1]\exp\left(\widehat{P}_i^{ul}[t{-}1]\right)\big]\Big]_{\log P^{i}_{0}\leq \widehat{P}_i^{ul}[t{-}1]\leq \log P_i^{max}}.\label{eq:UP_iteration}
\end{align}
  \State \textbf{Price update.} The prices $q_j^{dl}[t], j\in \mathcal{K}_{dl}$ are updated as follows,
  \begin{align}
q_j^{dl}[t]&=\Big[q_j^{dl}[t-1]+\gamma\big[r_j[t-1]-\log\big(\mathsf{SINR}_j[t-1]\big)\big]\Big]_{q_j^{dl}[t]\geq 0},\label{DL_Price_Update}
\end{align}
where $r_j[t]=(U'_j)^{-1}(q_j^{dl}[t])$. The uplink user $i\in \mathcal{K}_{ul}$ update its price by
  \begin{align}
q_i^{ul}[t]&{=}\Big[q_i^{ul}[t-1]{+}\gamma\big[r_i[t{-}1]{-}\log\big(\mathsf{SINR}_i[t{-}1]\big)\big]\Big]_{q_i^{ul}[t]\geq 0},\label{DL_Price_Update}
\end{align}
where $r_i[t]=(U'_i)^{-1}(q_i^{ul}[t])$.
  \State $t\rightarrow t+1$. Go to step 2.
\end{algorithmic}
\end{algorithm}

In the algorithm, the up- and downlink power is updated according to the gradient ~(\ref{eq:UL_gradient})(\ref{eq:DL_gradient}) in item 2 and 3. The price update in item 4 is along the direction of gradient of dual function $D(\mathbf{q}^{ul},\mathbf{q}^{dl})$. Also note that in uplink control algorithm~(\ref{eq:UP_iteration}) in item 3, the term 
\begin{align*}
m_{ij}[t-1]\exp\left(\widehat{P}_i^{ul}[t{-}1]\right)=q^{dl}_j[t-1]\frac{g^I_{ij}P^{ul}_i[t{-}1]}{\mathsf{IN}_j[t-1]}
\end{align*}
 denotes the contribution of inter-node interference from uplink user $i$ to the overall interference experienced at downlink user $j$, weighted by the price of user $j$. Therefore if user $i$ contributes significant interference to downlink user $j$ whose achievable rate is already very `pricey', user $i$'s power will tend to decrease.

Note that the above power control algorithm is implemented distributively at the uplink and downlink users. The distributed nature of the algorithm is facilitated by the decoupling property between up- and downlinks previously discussed, which is captured by the one-hop information. Specifically, because of the decoupling, given the price $\mathbf{q}^{dl}$, downlink power control is implemented independently from uplink users. Hence BS does not require uplink power information. Uplink user power control is implemented autonomously without awareness of the downlink power control policy. Note that the the one-hop information takes advantage of such structure so that all the network devices has sufficient information to perform optimal power control.

The following theorem establishes the convergence of the algorithm. The proof follows the similar lines as \cite{chiang2005balancing} and is hence neglected.
\newtheorem{theorem}{Theorem}
\begin{theorem}
Let $\phi^*$ be the set of $(\mathbf{\widehat{P}}^{ul}, \mathbf{\widehat{P}}^{dl})$ that optimally solves the problem ~(\ref{eq:obj_log})-(\ref{eq:ul_pwr_bound}). For small enough $\gamma$, we have $\lim_{t\rightarrow\infty}(\mathbf{\widehat{P}}^{ul}[t], \mathbf{\widehat{P}}^{dl}[t])\in \phi^*$.
\end{theorem}

It also follows from \cite{chiang2005balancing} that the algorithm converges to optimum geometrically, given in the following theorem, where $(\mathbf{P}^{ul,*},\mathbf{P}^{dl,*})\in \phi^*$ and
\begin{align*}
\epsilon[k]{=}\Big|\sum_{i\in\mathcal{K}_{ul}}U_i(\widetilde{R}_i^{ul}(P_i^{ul}[k])){+}\sum_{j\in\mathcal{K}_{dl}}U_i(\widetilde{R}_j^{dl}({P}_j^{dl}[k],\mathbf{P}^{ul}[k]))\nonumber\\
{-}\sum_{i\in\mathcal{K}_{ul}}U_i(\widetilde{R}_i^{ul}(P_i^{ul,*})){+}\sum_{j\in\mathcal{K}_{dl}}U_i(\widetilde{R}_j^{dl}({P}_j^{dl,*},\mathbf{P}^{ul,*}))\nonumber\Big|.
\end{align*}

\begin{theorem}
There exist constants $c$ and $\beta$ such that for all $k$, $\epsilon[k]\leq c \beta^k$.
\end{theorem}



\section{Power Control Implementation \\using overhearing-based scheme}
\label{sec:overhearing}

We have previously seen how the one-hop information architecture facilitates design of distributed algorithm for managing inter-node interference. Next, we present how the architecture can be implemented using an overhearing-based scheme.


At the end of each slot $t$, a feedback message is transmitted from each downlink user $j\in\mathcal{K}_{dl}$ to the BS. The feedback message $\mathsf{\mathsf{fb}}_j[t]$ from the $j$-th downlink user contains the following value,
\begin{align}
\mathsf{\mathsf{fb}}_j[t]{=}\frac{q^{dl}_j[t]}{\mathsf{IN}_j[t]}{=}\frac{q^{dl}_j[t]\cdot \mathsf{SINR}^{dl}_j[t]}{M P_j^{dl}[t]g^{dl}_j},\label{eq:feedback}
\end{align}
where $\mathsf{IN}_j[t]$ is the interference-plus-noise measured at the $j$-th downlink user. Note also that downlink user $j$ has all the knowledge to keep a copy of price $q^{dl}_j[t]$ according to (\ref{DL_Price_Update}).

The feedback message $\mathsf{fb}_j[t]$ is received at the BS and is overheard by all the uplink users in the neighborhood $N_j$ of $j$-th downlink user. We assume these feedback packets are sent with low rate and can be received collision-free at BS and users via orthogonal frequency/time resources. The BS and the uplink users then process the packets as follows.

$\bullet$ From the embedded pilot symbols in the feedback message $\mathsf{fb}_j[t]$, the BS estimates the downlink path loss $g^{dl}_{j}$. The BS then uses $\mathsf{fb}_j[t]$ to calculate $\mathsf{SINR}^{dl}_j[t]$ at the downlink receiver $j$  at time $t$, i.e.,
\begin{align}
\mathsf{SINR}^{dl}_j[t]=\frac{\mathsf{fb}_j[t]M P_j^{dl}[t]g^{dl}_j}{q^{dl}_j[t]}. \label{eq:BS_SINR}
\end{align}

$\bullet$ The uplink users in the neighborhood $\mathcal{N}_j$ of downlink user $j$ \emph{overhear} the feedback message. From the embedded pilot symbols in the feedback message $\mathsf{fb}_j[t]$, the uplink transmitter $i \in\mathcal{N}_j$ estimates the interference channel gain $g^I_{ij}$ associated with downlink receiver $j\in\mathcal{K}_{dl}$. The uplink transmitters $i \in\mathcal{N}_j$ then, from (\ref{eq:metric}) and (\ref{eq:feedback}), use the decoded feedback $\mathsf{fb}_j[t]$ to calculate the metric $m_{ij}[t]=\mathsf{fb}_j[t]P^{ul}[t]g^I_{ij}$ and update the power allocation according to~(\ref{eq:UP_iteration}).

Therefore, with the overhearing-based feedback scheme, the BS obtains the required SINR information in Table~1, and each uplink user obtains the IN metric in Table~2. Note that in Table~2, uplink SINR can be obtained from~(\ref{eq:SINR_UL}) where the uplink power is known to the user and path loss coefficient can be estimated via BS broadcasting pilots at initialization to all users.

\section{Asymptotic Performance of \\Power Control}
\label{sec:asymp}

Massive MIMO full-duplex system brings a large amount of possible degree of freedoms, hence facilitates the capability to simultaneously serve a large number of uplink and downlink users. In this section, we study the performance gain associated with power control, as well as the asymptotic scaling of the optimal power allocation in massive MIMO full-duplex system when both the number of antennas and the number of users scale up.

We hence consider a \text{sequence of scenarios} with expanding number of users and antennas, indexed by $l$ with $l\in\mathbb{Z}^+$. For the $k$-th scenario, the set of uplink users, downlink users, and the number of BS antennas are denoted as $\mathcal{K}_{ul}^{k}$, $\mathcal{K}_{dl}^{k}$, $M_k$, respectively.

To study the system behavior when more and more users are added to the system, we study the monotonic scenario where $M_l<M_{l+1}$, and the uplink and downlink set of users monotonically expand across scenarios, i.e., $\mathcal{K}_{ul}^1\subsetneq \mathcal{K}_{ul}^2\subsetneq \mathcal{K}_{ul}^3\subsetneq\cdots$ and $\mathcal{K}_{dl}^1\subsetneq \mathcal{K}_{dl}^2 \subsetneq\mathcal{K}_{dl}^3\subsetneq\cdots$. We hence have $|\mathcal{K}_{ul}^{l}|\rightarrow\infty$, $|\mathcal{K}_{dl}^{l}|\rightarrow\infty$ and $M_l\rightarrow\infty$ as $l\rightarrow\infty$. 

We assumes the inter-node interference $g^{I}_{i.j}, i\in\mathcal{K}_{ul}, j\in\mathcal{K}_{dl}$ follow certain distribution with expectation $\mathbb{E}[g^{I}_{i,j}]=\mathcal{E}$ uniformly across $(i,j)$ pairs. This assumption is valid for the typical scenarios where users randomly appear in geographical locations of a cell, or where inter-node interference arrives via random paths.

We assume the following conditions hold true for the utility functions of the users in the network: 

(A1) For each user $k$, $\lim_{r\rightarrow \infty} U_k'(r)=0$, $k\in \mathcal{K}_{ul}\cup\mathcal{K}_{dl}$. 

(A2) The utility function $U_i$ of $i$-th user satisfies $U_i\in \mathcal{U}$, where $\mathcal{U}=\{U^1(r), U^2(r), \cdots, U^M(r)\}$ is a finite collection of utility functions.

Assumption (A1) is motivated by the diminishing return idea of the utility functions, and holds true for the well-known $\alpha$-fairness, e.g, 
\begin{align*}
U_i(r)=\omega_i \frac{r^{1-\alpha}}{1-\alpha}, \alpha>0, \alpha\neq 1, \omega_i> 0,
\end{align*}
which includes many well-known fairness criteria such as proportional fairness, minimum
potential delay fairness and max-min fairness as special cases that corresponds to different values of $\alpha$. Assumption (A2) states that there are finite choices of utility functions for each user. These choices correspond to, for instance, different types of data traffic, quality of service requirement, etc. 

For ease of exposition, we henceforth use the following notations:

\vspace{5pt}$\bullet$ Vector $\mathbf{P}^{ul,l,*}$ and $\mathbf{P}^{dl,l,*}$  respectively denote the optimal uplink and downlink power allocation that maximizes the network sum-utility (i.e.,~(\ref{eq:obj_log})-(\ref{eq:ul_pwr_bound})) in the $l$-th scenario.  

\vspace{5pt}$\bullet$ The utility $U_i^{l,*}$, $i\in \mathcal{K}^l_{ul}\cup \mathcal{K}^l_{dl}$ respectively represents the up- and downlink utility under optimal power allocation $\mathbf{P}^{ul,l,*}$ and $\mathbf{P}^{dl,l,*}$ in the $l$-th scenario.

\vspace{5pt}$\bullet$ The utility $U_i^l({P}^{ul}_{max}), i{\in}\mathcal{K}^l_{ul}$, and $U_j^l(\mathbf{P}^{ul}_{\max}), j{\in}\mathcal{K}^l_{dl}$ represent the utility when each uplink user transmits at maximum power ${P}^{ul}_{\max}$, and when the downlink users transmits at optimal power that maximizes the sum downlink utility under $\mathbf{P}^{ul}_{max}$. 

\vspace{8pt}The next theorem states that when the number $M_l$ of the BS antennas scales at the same order of the multiplication of $|\mathcal{K}^l_{ul}|$ and $|\mathcal{K}^l_{dl}|$, and when the number of downlink users scales at the same order or faster than the number of uplink users, then asymptotically the optimal uplink transmission power scales down.
\begin{theorem}\label{theorem:linear}
Suppose $\lim_{l\rightarrow\infty} \frac{M_l}{|\mathcal{K}^l_{ul}||\mathcal{K}^l_{dl}|}= C$ for some constant $C>>1$, and $\limsup_{l\rightarrow\infty}\frac{|\mathcal{K}^l_{ul}|}{|\mathcal{K}^l_{dl}|}< \infty$. For any $0<\rho<1$, we let $\Theta^l_\rho\subseteq \mathcal{K}^l_{ul}$ be the subset of uplink users with $P^{ul,l,*}_i\leq \rho P^{ul}_{max}$. Then $\lim_{l\rightarrow\infty} \frac{|\Theta_{\rho}^l|}{|\mathcal{K}^l_{ul}|}=1$.
\end{theorem}

\begin{proof}
See Appendix A.
\end{proof}

\textbf{Remark:} 

(1) Theorem~\ref{theorem:linear} studies the asymptotic regime when $\lim_{l\rightarrow\infty} \frac{M_l}{|\mathcal{K}^l_{ul}||\mathcal{K}^l_{dl}|}= C$, where the BS antennas $M_l$ scales at the same order of the multiplication of $K^l_{ul}$ and $K^l_{dl}$. Intuitively this means the amount of antennas grows sufficiently fast to combat \emph{simultaneously} the growing inter-node interference (from the growing number of uplink users) and also to provide sufficient degrees of freedom to the downlink users to guarantee a non-trivial SINR gain at the downlink users expressed in~(\ref{eq:SINR_DL}). Theorem~\ref{theorem:linear} also assumes that $\limsup_{l\rightarrow\infty}\frac{|\mathcal{K}^l_{ul}|}{|\mathcal{K}^l_{dl}|}< \infty$, i.e., in the asymptotic regime, the number of downlink users scales at the same order or faster than the number of uplink users is valid in typical scenarios where the network is often dominated by the downlink traffic.

(2) Theorem~\ref{theorem:linear} states that in the considered asymptotic regime, the optimal transmission power asymptotically scales down. This is intuitive because the uplink SINR grows linearly with the number of antennas. However, because of the diminishing decay property of the utility function, the marginal utility gain due to the increasing number of antennas will reduce. Therefore, decreasing the uplink transmission power (e.g., decrease power by $\rho$) will not bring significant loss in the uplink utility, it does, however, bring significant downlink utility gain. This is because, as the number of uplink users increases, the aggregated interference can be large at \emph{each} of the downlink users.

\begin{figure}
\centering
\includegraphics[width=1.65in]{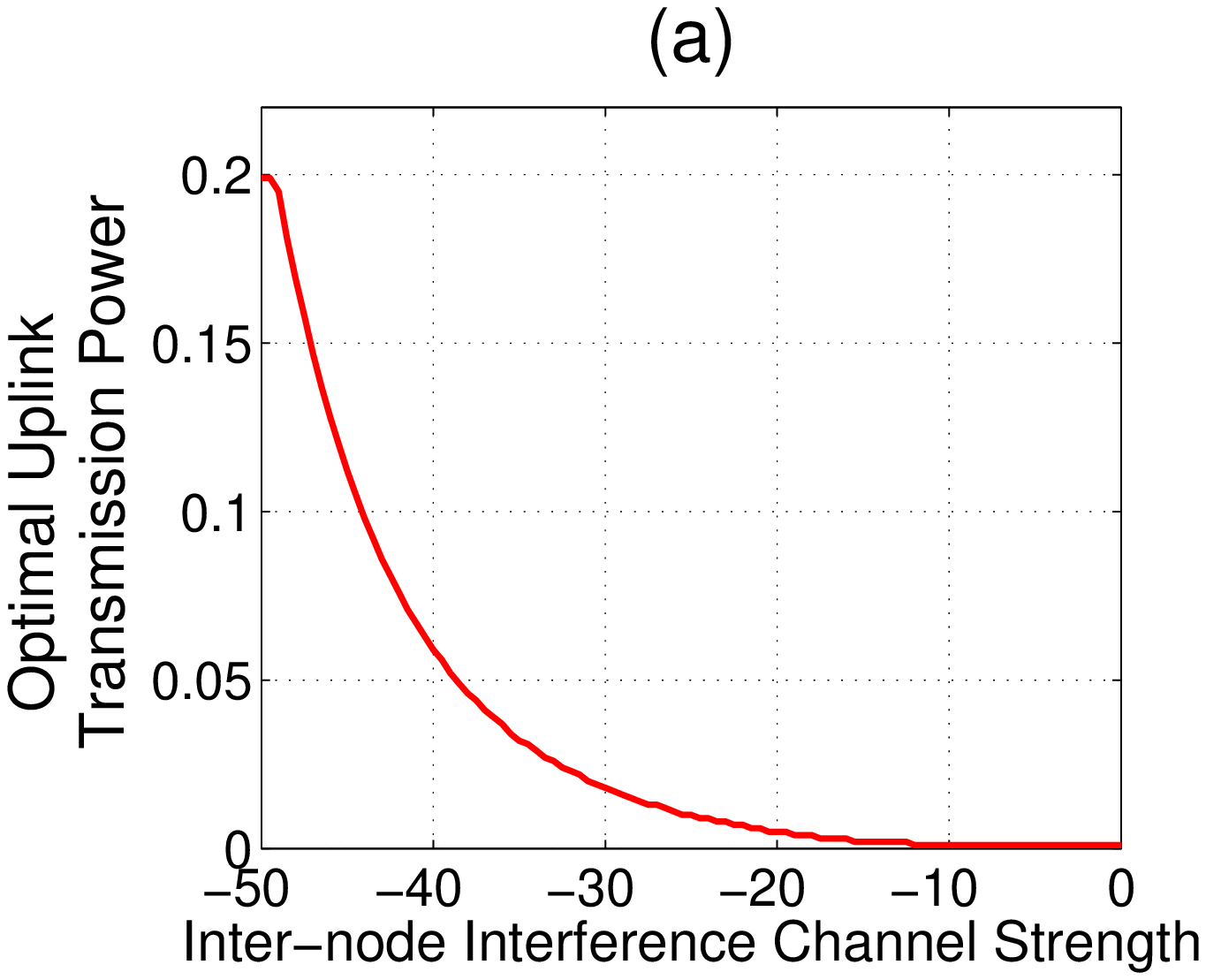}
\includegraphics[width=1.62in]{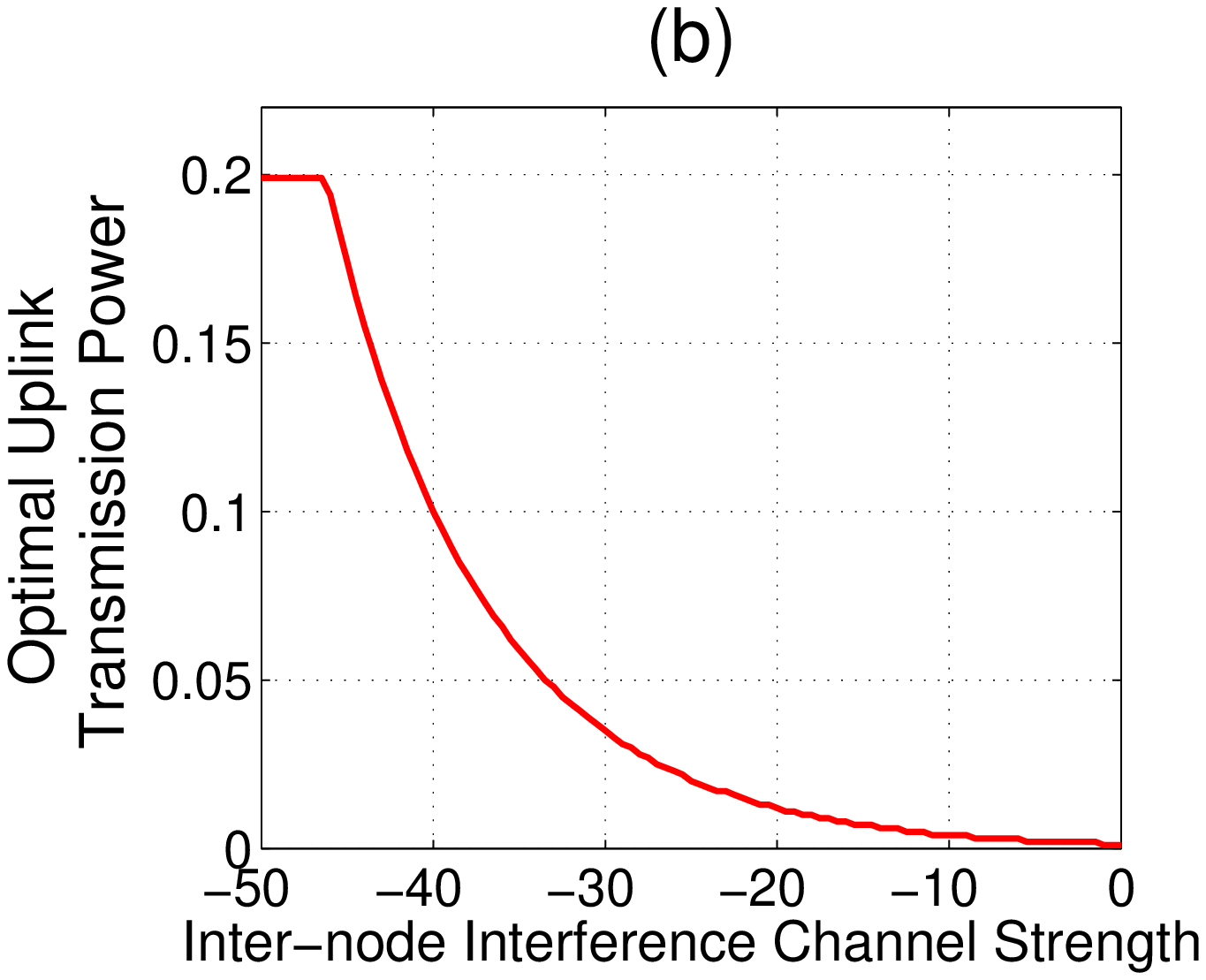}
\includegraphics[width=1.64in]{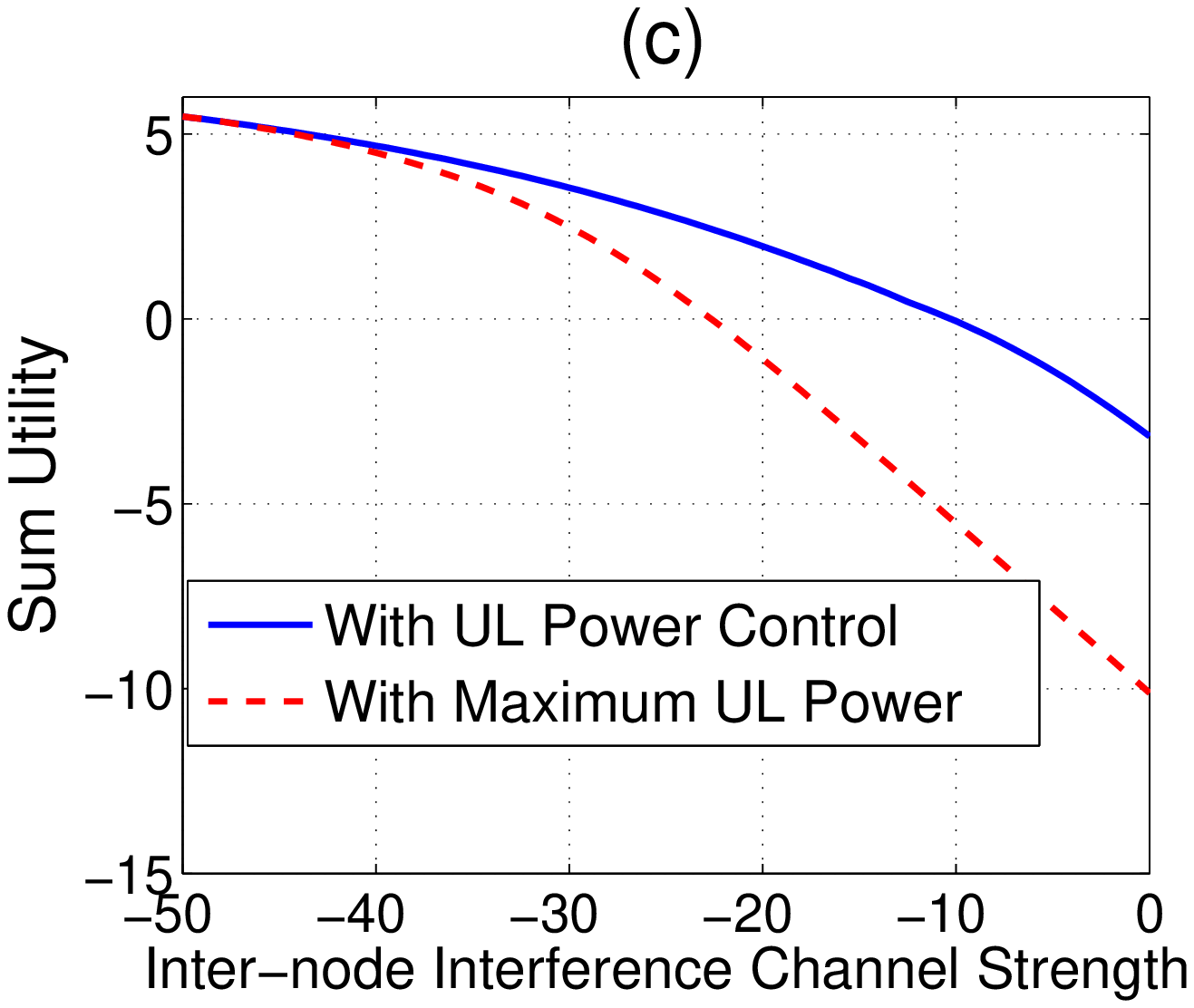}
\includegraphics[width=1.64in]{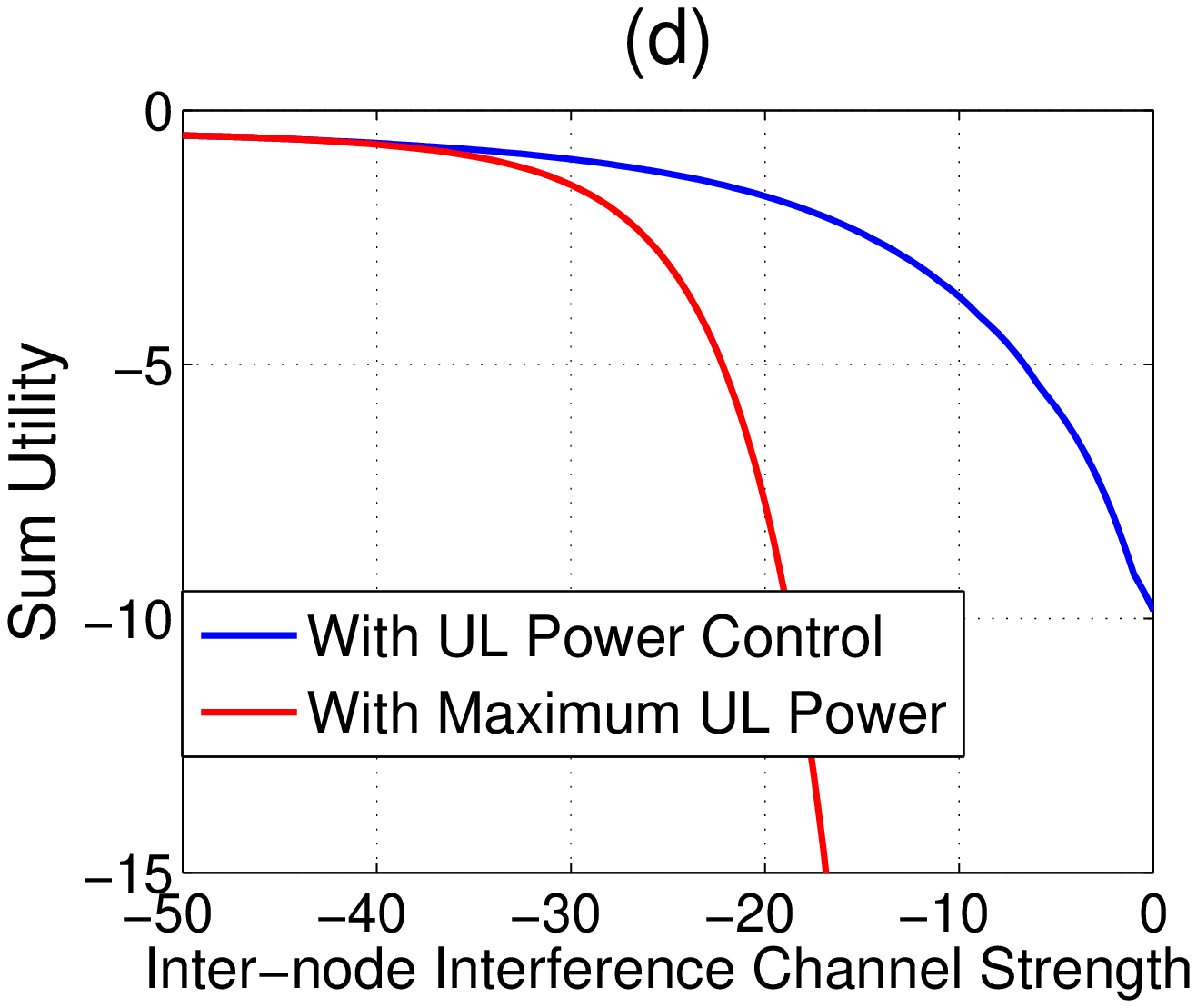}
\vspace{-7.5pt}
\caption{Performance evaluation of power control algorithm for two-user case. (a) Optimal power scaling with weighted proportional fair utility; (b) Optimal power scaling with weighted minimum potential delay utility;(c) Comparison with Naive scheme with $P^{ul}=P^{ul}_{max}$ under weighted proportional fair utility; (d) Comparison with Naive scheme with $P^{ul}=P^{ul}_{max}$ under weighted minimum potential delay utility.}
\label{fig:Two-user}
\vspace{-8pt}
\end{figure}





\section{Numerical Evaluations}
\label{sec:numerics}

We examine the power control algorithm performance for the massive MIMO full-duplex with one uplink user and one downlink user. We consider the scenario with $128$ antennas at the BS, i.e., $M=128$. The uplink and downlink path loss parameters are $g^{ul}=-60~dB$, $g^{dl}=-70~dB$ respectively. The power control is subject to the maximum power constraints $P^{ul}_{max}= 23~dBm$, $P^{dl}_{tot}= 45~dBm$. The noise power is $N_0=-30~dBW$. These parameters are set in accordance with practical values in existing LTE standards. Fig.~\ref{fig:Two-user}(a) and Fig.~\ref{fig:Two-user}(c) correspond to weighted proportional fair utility function, with $U^{ul}(r)=\log(r), U^{dl}(r)=2\log(r)$.  Figure~\ref{fig:Two-user}(b) and Figure~\ref{fig:Two-user}(d) correspond to weighted minimum potential delay utility function, i.e., $U^{ul}(r)=-1/r, U^{dl}(r)=-2/r$. For this scenario, it is optimal for the downlink user to always transmit at the maximum power since there is only one downlink user.

Figure~\ref{fig:Two-user}(a)-(b) plot the variation of the optimal uplink transmission power with the gain $g^I$ of the inter-node interference channel. Figure~\ref{fig:Two-user}(c)-(d) compare the sum-utility under optimal power control and when $P^{ul}=P^{ul}_{max}$. As observed in the figures, when the inter-node interference channel is weak, it is optimal to transmit at the maximum uplink power since the impact on downlink transmission is insignificant. Hence in this regime the uplink transmits power at the maximum power and there is a negligible loss at the downlink as evident in Figure~\ref{fig:Two-user}(c)-(d). As the inter-node interference channel gain grows stronger, for both cases of utility functions, the optimal uplink transmission power decays. This is because under strong interference channel, transmission of uplink users imposes significant interference to the downlink transmission.  Because of the diminishing return property of utility functions, reducing the uplink transmission power can alleviate downlink utility loss and bring higher overall sum utility in the network. 

Figure~\ref{fig:Two-user}(c)-(d) also highlight the importance of power control in moderate to high interference regime. As the gain in the inter-node interference channel becomes stronger, power control is increasingly more important to manage the interference created by the uplink users to the downlink receivers to maintain the overall system-level utility, as illustrated by the increasing utility gains in Figure~\ref{fig:Two-user}(c)(d).

\begin{figure}
\centering
\includegraphics[width=3.1in]{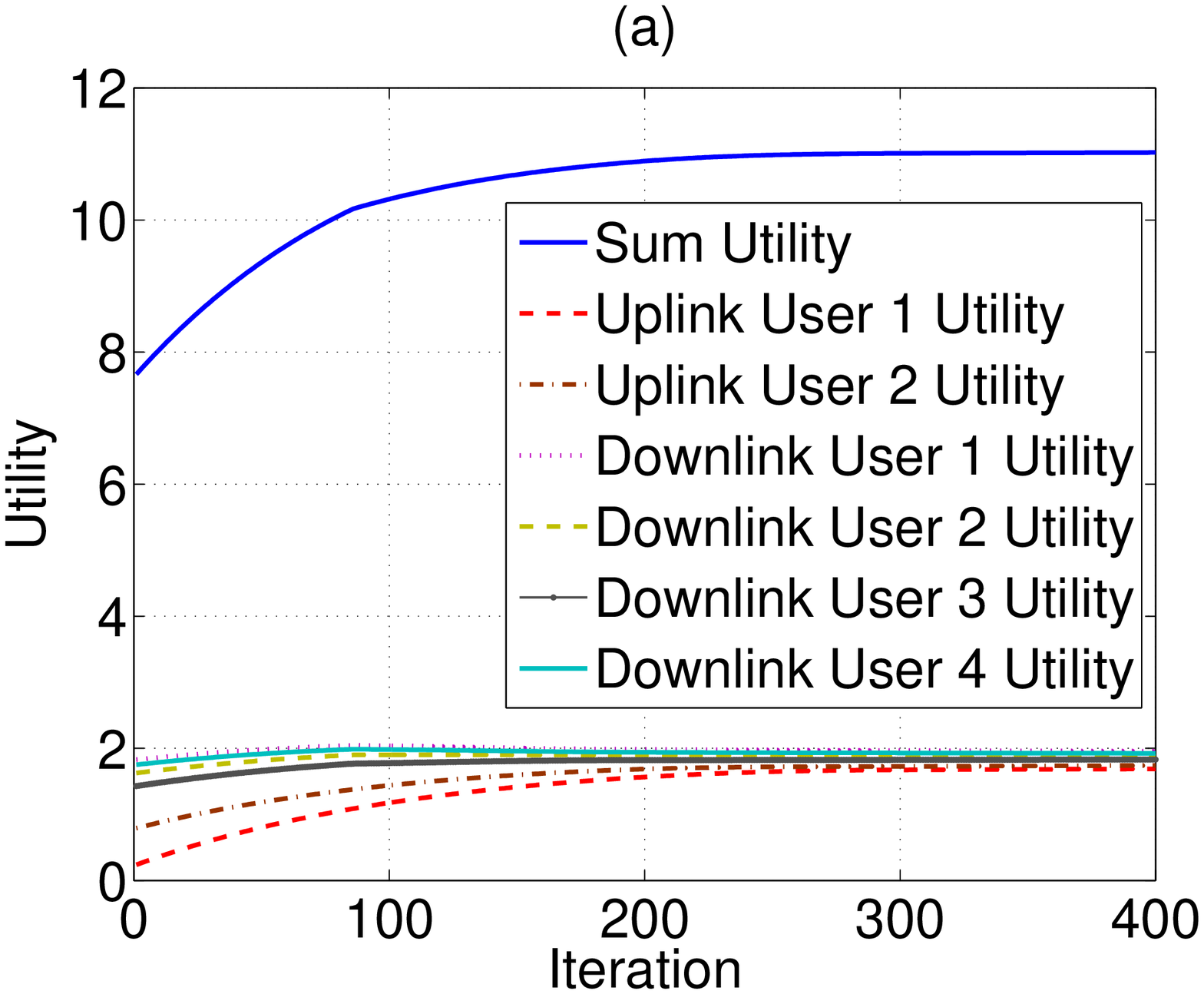}
\includegraphics[width=3.1in]{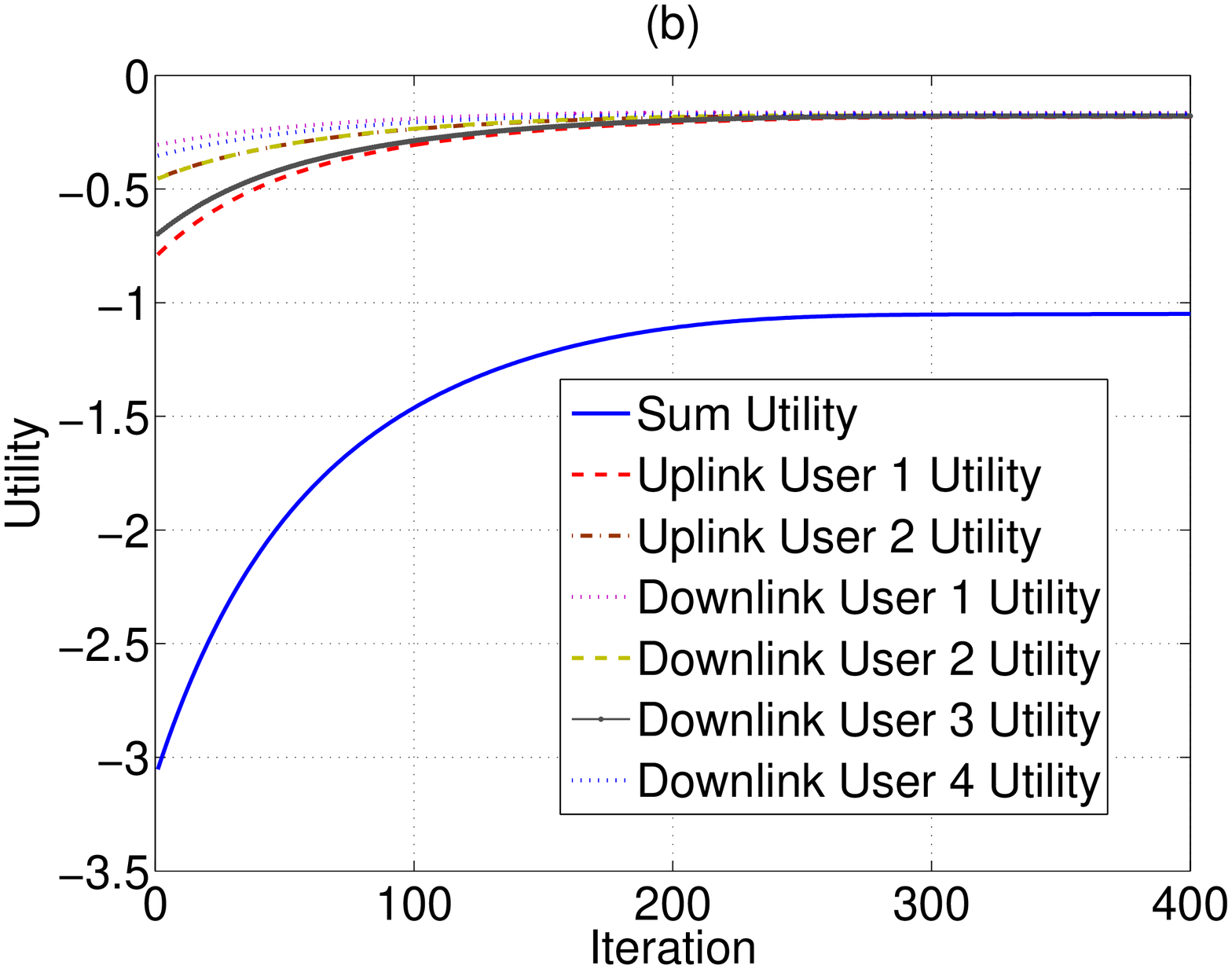}
\vspace{-7.5pt}
\caption{Convergence of power control algorithm. (a) Proportional fair utility; (b) Minimum potential delay utility.}
\label{fig:ROC}
\vspace{-8pt}
\end{figure}

Fig~\ref{fig:Two-user}(a)-(b) plot convergence performance of the proposed algorithm. We consider the scenario with $128$ antennas at the BS, with $2$ uplink users and $4$ downlink users. The uplink and downlink path loss parameters are $g^{ul}_1=-50 dB, g^{ul}_2=-45 dB, $ and $g^{dl}_1=-56 dB, g^{dl}_2=-61 dB, g^{dl}_3=-65 dB, g^{dl}_4=-58 dB$. The inter-node interference channel gain is $g^I_{1,1}=-59 dB, g^I_{1,2}=-60 dB, g^I_{2,1}=-62 dB, g^I_{2,2}=-55 dB$. Fig~\ref{fig:Two-user}(a) plots the evolution of the utility function for proportional fairness utilities, i.e., $U_k(r)=\log r, k\in \mathcal{K}_{ul}\cup\mathcal{K}_{dl}$. Fig~\ref{fig:Two-user}(b) plots the evolution of the utility function for minimum potential delay utility, i.e., $U_k(r)=-1/r, k\in \mathcal{K}_{ul}\cup\mathcal{K}_{dl}$. It can be observed that the algorithm converges at around $200$ iterations.

\section{Conclusion}
\label{sec:conclusion}
In this paper, we study a massive MIMO full-duplex system where the BS contains a large number of full-duplex-capable antennas serving multiple single-antenna half-duplex up and -downlink users. In such system, inter-node inference becomes the main design challenge, where transmission of uplink users creates interference to the reception of downlink users. Because of the vast amount of inter-node interference channels, centralized control by aggregating all network information at the BS will incur significant overhead in the network. We develop an efficient and scalable one-hop information architecture that leverages the unique structure of the massive MIMO full-duplex system. Based on the one-hop information architecture, we propose a distributed power allocation algorithm to optimally manage the inter-node interference, and show how the information can be obtained via a overhearing-based scheme. 
We characterize the performance of the optimal power control algorithm in the asymptotic regimes where the number of users and the number of antennas at the BS all scale up. We show that power control in large number of BS antennas regime can not only bring utility gains, it also improves the energy efficiency of mobile users.

\bibliographystyle{abbrv}
\bibliography{Mobihoc2015}  
%
%

\appendices
\section{Proof of Theorem~2}

We let $\Phi_{\rho}^l=\mathcal{K}^l_{ul}\setminus\Theta_{\rho}^l$. Hence for all $i\in \Phi_{\rho}^l$, $P^{ul,l,*}_i>\rho P^{ul}_{max}$. We also let $\tilde{U}'(r)=\max_{m}\{(U^m)'(r),U^m\in\mathcal{U}\}$. 

We prove this lemma by contradiction. Suppose there exist a subsequence $l_m, m=1,2,\cdots$, and $0<\kappa\leq 1$ such that $\lim_{m\rightarrow\infty} |\Phi_{\rho}^{l_m}|/\mathcal{K}_{ul}=\kappa$. We let $\mathbf{P}^{ul,l_m,new}$ be the power allocation vector where $P^{ul,l_m,new}_i=\rho_0 P^{ul}_{max}, i\in \Phi_{\rho}^{l_m}$ with $\rho_0{<}\rho$, and ${P}^{ul,l_m,new}_i=P^{ul,l_m,*}_{i}, i\in \Theta^{l_m}_{\rho}$. For uplink users $i\in\Phi_{\rho}^{l_m}$ and $m$ being large, the achievable transmission rate when $P^{ul}_i=P^{ul,l_m,new}_{i}$ is expressed as follows,
\begin{align}
&R_i^{ul}(P^{ul,l_m,new}_{i},M_{l_m})
=\log\Big(\frac{M_{l_m}\cdot \rho_0 P^{ul}_{max} g_i^{ul}}{N_0}\Big)\nonumber\\
=& \log\Big(\frac{M_{l_m}\cdot P^{ul,l_m,*}_i g_i^{ul}}{N_0}\Big){+}\log (\frac{M_{l_m}\cdot  \rho_0 P^{ul}_{max}g_i^{ul}}{M_{l_m}\cdot P^{ul,l_m,*}_i g_i^{ul}})\nonumber\\
=&R_i^{ul}(P^{ul,l_m,*}_{i},M_{l_m})+\log g_i({l_m}).\label{eq:UL_opt_rate}
\end{align}
where $g_i(l)>{\rho_0}.$ From~(\ref{eq:UL_opt_rate}) and concavity of utility function, we have for $i\in\Phi^{l_m}$
\begin{align}
& U_i(R_i^{ul}(P^{ul,l_m,new}_{i},M_{l_m}))-U_i(R_i^{ul}(P^{ul,l_m,*}_{i},M_{l_m}))\nonumber\\
\geq& U_i'(R_i^{ul}(P^{ul,new}_{i},M_{l_m}))\cdot \log g_i(l_m). \label{eq:up_Ui_diff}
\end{align}

For the downlink users, the achievable transmission rate satisfies
\begin{align}
&R_j^{dl}(P^{dl,l_m,*}_{j},\mathbf{P}^{ul,l_m,new})=\log\Big(\frac{M\cdot P^{dl,l_m,*}_j g_j^{dl}}{\sum_{i=1}^{K_{ul}}g^I_{ij}P^{ul,new}_i+N_0}\Big)\nonumber\\
=&\log\Big(\frac{M\cdot P^{dl,l_m,*}_j g_j^{dl}}{\sum_{i\in\Phi^{l_m}}g^I_{ij}\rho_0 P^{ul}_{max}+\sum_{i\in\Theta^{l_m}}P^{ul,l_m,*}_i g^I_{ij}+N_0}\Big)\nonumber\\
= &R_j^{dl}\big(P^{dl,l_m,*}_{j}, \mathbf{P}^{ul,l_m,*}){+}\log f_j(l_m),\label{eq:DL_opt_rate}
\end{align}
where
\begin{align*}
f_j(l_m)=&\frac{\sum_{i\in\mathcal{K}_{ul}^{l_m}}P^{ul,l_m,*}_i g^I_{ij}}{\big(\sum_{i\in\Phi^{l_m}}g^I_{ij}\rho_0 P^{ul}_{max}+\sum_{i\in\Theta^{l_m}}P^{ul,l_m,*}_i g^I_{ij}+N_0\big)}\nonumber\\
=&1+\frac{\sum_{i\in\Phi^{l_m}}g^I_{ij}P^{ul,l_m,*}_i-\sum_{i\in\Phi^{l_m}}g^I_{ij}\rho_0 P^{ul}_{max}-N_0}{\sum_{i\in\Phi^{l_m}}g^I_{ij}\rho_0 P^{ul}_{max}+\sum_{i\in\Theta^{l_m}}P^{ul,l_m,*}_ig^I_{ij}+N_0},
\end{align*}
which satisfies,
\begin{align*}
f_j(l_m)\geq&1+\frac{\big(\sum_{i\in\Phi^{l_m}_{\rho}}g^I_{ij}\big[\rho_l-\rho_0\big]P^{ul}_{max}\big)-N_0}{\big(\sum_{i\in\Phi^{l_m}_{\rho}}g^I_{ij}\rho P^{ul}_{max}+\sum_{i\in\Theta^{l_m}_{\rho}}P^{ul,l_m,*}_i+N_0}\nonumber\\
=&1+\frac{\big[\rho-\rho_0\big]P^{ul}_{max}\sum_{i\in\Phi^{l_m}_{\rho}}g^I_{ij}-N_0}{\big(\sum_{i\in\Phi^{l_m}_{\rho}}g^I_{ij}\rho P^{ul}_{max}+\sum_{i\in\Theta^{l}}P^{ul,l_m,*}_ig^I_{ij}\big)+N_0}.\nonumber
\end{align*}

Hence
\begin{align*}
\limsup_{m\rightarrow\infty}f_j(l_m)\geq &1+\limsup_{m\rightarrow\infty}\frac{\big[\rho-\rho_0\big]P^{ul}_{max}\sum_{i\in\Phi^{l_m}}g^I_{ij}/|\mathcal{K}_{ul}^{l_m}|}{\big(\sum_{i\in\mathcal{K}_{ul}^{l_m}}g^I_{ij}\rho P^{ul}_{max}\big)/|\mathcal{K}_{ul}^{l_m}|}\nonumber\\
\geq& 1+\frac{\big[\rho-\rho_0\big]C_{\kappa}}{\rho\mathbb{E}[g^I_{ij}]}.\nonumber
\end{align*}
where $C_{\kappa}=\mathbb{E}[g^I_{ij} | g^I_{ij}\leq \alpha]$ with $\alpha$ being such that $P(g^I_{ij}\leq \alpha)=\kappa$. Next consider the difference in utility. From~(\ref{eq:DL_opt_rate}) and concavity of utility, we have for downlink user $j$,
\begin{align}
& U_j(R_j^{dl}({P}^{dl,l_m,*}_{j},\mathbf{P}^{ul,l_m,new}))-U_j(R_j^{dl}(P^{dl,l_m,*}_{j},\mathbf{P}^{ul,l_m,*}))\nonumber\\
\geq& U_j'(R_j^{dl,l_m,*}(P^{dl,l_m,*}_{j},\mathbf{P}^{ul,l_m,new}))\cdot \log f_j(l_m). \label{eq:up_Di_diff}
\end{align}

We let $\mathbf{P}^{dl,l_m,new*}$ denote the optimal downlink power allocation to maximize the sum-downlink-utility given uplink power allocation $\mathbf{P}^{ul,l_m,new}$. The difference in overall sum-utility is
\begin{align*}
&\Big[\hspace{-4pt}\sum_{\substack{i\in\mathcal{K}^{l_m}_{ul}}}\hspace{-4pt}U_i^{l_m,new}{+}\hspace{-4pt}\sum_{j\in\mathcal{K}_{dl}^{l_m}}\hspace{-4pt}U_j^{l_m,new,*}\Big]-\Big[\hspace{-4pt}\sum_{\substack{i\in\mathcal{K}^{l_m}_{ul}}}\hspace{-4pt}U_i^{l_m,*}{+}\hspace{-4pt}\sum_{j\in\mathcal{K}^{l_m}_{dl}}\hspace{-4pt}U_j^{l_m,*}\Big]\\
\geq& \Big[\hspace{-4pt}\sum_{\substack{i\in\mathcal{K}^{l_m}_{ul}}}\hspace{-4pt}U_i^{l_m,new}{+}\hspace{-4pt}\sum_{j\in\mathcal{K}_{dl}}\hspace{-4pt}U_j^{l_m,new}\Big]-\Big[\hspace{-4pt}\sum_{\substack{i\in\mathcal{K}^{l_m}_{ul}}}\hspace{-4pt}U_i^{l_m,*}{+}\hspace{-4pt}\sum_{j\in\mathcal{K}^{l_m}_{dl}}\hspace{-4pt}U_j^{l_m,*}\Big]\\
=& \Big[\hspace{-4pt}\sum_{\substack{i\in\mathcal{K}^{l_m}_{ul}}}\hspace{-4pt}U_i^{l_m,new}-U_i^{l_m,*}\Big]{+}\Big[\sum_{j\in\mathcal{K}^{l_m}_{dl}}\hspace{-4pt}U_j^{l_m,new}-U_j^{l_m,*}\Big]\\
\geq& \sum_{\substack{i\in\Phi^{l_m}}}U_i'(R_i^{ul}(P^{ul,l_m,new}_{i},M_{l_m}))\cdot \log g_i(l_m)\\
&\hspace{0.2in}+\sum_{j\in\mathcal{K}^{l_m}_{dl}}U_j'(R_j^{dl,l_m,*}((P^{dl}_{j},\mathbf{P}^{ul,l_m,new}))\cdot \log f_j(l_m).
\end{align*}
where the last inequality holds from~(\ref{eq:up_Ui_diff})(\ref{eq:up_Di_diff}). We then have
\begin{align*}
&\liminf_{m\rightarrow\infty}\Big[\sum_{\substack{i\in\Phi^{l_m}}}U_i'(R_i^{ul}(P^{ul,l_m,new}_{i},M_{l_m}))\cdot \log g(l_m)\\
&\hspace{0.2in}+\sum_{j\in\mathcal{K}^{l_m}_{dl}}U_j'(R_j^{dl}((P^{dl,l_m,*}_{j},\mathbf{P}^{ul,l_m,new}))\cdot \log f_j(l_m)\Big]\\
\geq&  \liminf_{m\rightarrow\infty}\Big[|{\Phi}^{l_m}_{\rho}|\widehat{U}'(R_i^{ul}(P^{ul,l_m,new}_{i},M_{l_m}))\log{\rho_0}\\
&\hspace{0.2in}+\sum_{j\in\mathcal{K}^{l_m}_{dl}}U_j'(R_j^{dl}((P^{dl,l_m,*}_{j},\mathbf{P}^{ul,l_m,new}))\cdot \log f_j(l_m)\Big],
\end{align*}
where recall that $\widehat{U}'(r)=\max_{i} {U}_i'(r)$. Note that 
\begin{align}
&\limsup_{m\rightarrow\infty} R_j^{dl}(P^{dl,l_m,*}_{j},\mathbf{P}^{ul,l_m,new})\nonumber\\
=&\limsup_{m\rightarrow\infty}\log\Big(\frac{M\cdot P^{dl,l_m,*}_j g_j^{dl}}{\sum_{i\in\Phi^{l_m}}g^I_{ij}\rho_0 P^{ul}_{max}+\sum_{i\in\Theta^{l_m}}P^{ul,l_m,*}_ig^I_{ij}+N_0}\Big)\nonumber\\
\leq& \limsup_{m\rightarrow\infty}\log\Big(\frac{M\cdot P^{dl,l_m,*}_j g_j^{dl}}{\sum_{i\in\Phi^{l_m}}g^I_{ij}\rho_0 P^{ul}_{max}+\sum_{i\in\Theta^{l_m}}P^{ul,l_m,*}_i g^I_{ij}+N_0}\Big)\nonumber\\
\leq& \limsup_{m\rightarrow\infty}\log\Big(\frac{M\cdot P^{dl,l_m,*}_j g_j^{dl}/|\mathcal{K}^{l_m}_{ul}|}{\big(\sum_{i\in\Phi^{l_m}}g^I_{ij}\rho_0 P^{ul}_{max}+N_0\big)/|\mathcal{K}^{l_m}_{ul}|}\Big)\nonumber\\
=& \limsup_{m\rightarrow\infty}\log\Big(\frac{M\cdot P^{dl,l_m,*}_j g_j^{dl}/|\mathcal{K}^{l_m}_{ul}|}{C_{\kappa}\mathbb{E}[g^I_{ij}]\rho_0 P^{ul}_{max}}\Big),\label{eq:limit_half}
\end{align}
where recall that $C$ was defined in the statement of the theorem, and $C_{\kappa}$ was previously defined in this proof. We proceed with the following lemma.
\newtheorem{lemma}{Lemma}
\begin{lemma}
Let $\Psi_{\omega}^l=\{j\in\mathcal{K}_{dl}^l: P_j^{dl,l,*}<\omega\frac{P^{ul}_{tot}}{|\mathcal{K}_{dl}^l|}\}$ for $\omega>0$. There exists $\omega>0$ such that $\limsup_{l\rightarrow\infty}\frac{|\Psi_{\omega}^l|}{|
\mathcal{K}_{dl}^l}>0$. 
\end{lemma}

\begin{proof}
Suppose the statement is not true. Then for any $\omega>1$ we have $\lim_{l\rightarrow\infty}\frac{|\Psi_{\omega}^l|}{|
\mathcal{K}_{dl}^l|}=0$. Therefore 
\begin{align*}
\lim_{l\rightarrow\infty}\sum_{j\in \mathcal{K}_{dl}^l}P_j^{dl,l,*}\geq\lim_{l\rightarrow\infty}\sum_{j\notin \Psi_{\omega}^l}P_j^{dl,l,*}\geq\omega P^{ul}_{tot}>P^{ul}_{tot}.\nonumber
\end{align*}
Hence the total downlink transmission power exceeds the power constraint, thus establishing a contradiction.
\end{proof}

According to the above lemma, we let $\omega_0$ and $\varepsilon$ be such that $\limsup_{m\rightarrow\infty}\frac{|\Psi_{\omega_0}^{l_m}|}{|
\mathcal{K}_{dl}^{l_m}|}=\varepsilon$. Therefore
\begin{align*}
&  \liminf_{m\rightarrow\infty}\Big[|{\Phi}^{l_m}_{\rho}|\widehat{U}'(R_i^{ul}(P^{ul,l_m,new}_{i},M_{l_m}))\log{\rho_0}\\
&\hspace{0.2in}+\sum_{j\in\mathcal{K}^{l_m}_{dl}}U_j'(R_j^{dl}((P^{dl,l_m,*}_{j},\mathbf{P}^{ul,l_m,new}))\cdot \log f_j(l_m)\Big],\\
\geq & \liminf_{m\rightarrow\infty}\Big[|{\Phi}^{l_m}_{\rho}|\widehat{U}'(R_i^{ul}(P^{ul,l_m,new}_{i},M_{l_m}))\log{\rho_0}\\
&\hspace{0.2in}+|\Psi_{\omega_0}^{l_m}|{U}_j'(\frac{C \omega_0 g_j^{dl}}{C_{\kappa}\mathbb{E}[g^I_{ij}]\rho_0})\cdot \log f_j(l_m)\Big],\\
=& K^{l_m}_{ul} \liminf_{m\rightarrow\infty}\Big[\kappa\widehat{U}'(R_i^{ul}(P^{ul,l_m,new}_{i},M_{l_m}))\log{\rho_0}\\
&\hspace{0.2in}+\omega_0\frac{|\mathcal{K}^{l_m}_{dl}|}{|\mathcal{K}^{l_m}_{ul}|}{U}_j'(\frac{C \omega_0 g_j^{dl}}{C_{\kappa}\mathbb{E}[g^I_{ij}]\rho_0})\cdot \log f_j(l_m)\Big],\\
>&0,
\end{align*}
where the last equality holds because the first term in the limit decays to $0$ with $l_m$ of~(\ref{eq:limit_half}), $\widehat{U}'(r), {U}'(r)$ monotonically decrease with $r$, the path loss $g_j^{dl}{<}1$, and because from the theorem statement $\liminf_{m\rightarrow\infty}\frac{|\mathcal{K}^{l_m}_{dl}|}{|\mathcal{K}^{l_m}_{ul}|}>0$.

Therefore choosing $(\mathbf{P}^{ul,l,new},\mathbf{P}^{dl,l,new*})$ asymptotically achieves utility higher than the optimal power allocation $(\mathbf{P}^{ul,l,*},\mathbf{P}^{dl,l,*})$, which contradicts to $(\mathbf{P}^{ul,l,*},\mathbf{P}^{dl,l,*})$ being the optimal power allocation. We hence proved the Theorem.

\end{document}